 \documentclass[conference,romanappendices,letterpaper]{IEEEtran}
 \pdfoutput=1 
\usepackage{hyperref}
\usepackage{subfig}
\usepackage{stfloats}
\IEEEoverridecommandlockouts
\newcommand{\nonl}{\renewcommand{\nl}{\let\nl\oldnl}}
\usepackage{etoolbox}
\AtEndEnvironment{example}{\null\hfill$\square$}%
\AtEndEnvironment{remark}{\null\hfill$\square$}%

\usepackage{enumerate}
\usepackage{multirow,array}
\usepackage{cite}
\usepackage{graphicx}
\usepackage{psfrag}
\usepackage{array}
\usepackage{algorithm2e}
\usepackage{amssymb}
\usepackage{amsfonts}
\usepackage{float}
\usepackage{textcomp}
\usepackage{xcolor}
\usepackage{tabu}
\usepackage{array}
\usepackage{tabularx}
\usepackage{subfiles}

\usepackage{amssymb}
\usepackage{pifont}
\usepackage[amsthm,thmmarks]{ntheorem}

\newcolumntype{P}[1]{>{\centering\arraybackslash}p{#1}}
\newcolumntype{M}[1]{>{\centering\arraybackslash}m{#1}}

 \usepackage{amsmath}

\usepackage{float}

\newtheorem{conjecture*}{Conjecture}
\newtheorem{lemma}{Lemma}
\newtheorem{observation}{Observation}
\newtheorem{corollary}{Corollary}
\newtheorem{claim}{Claim}
\newtheorem{remark}{Remark}

\newtheorem{theorem}{Theorem}

\title{Weak Private Information Retrieval for ~\\ Graph-based Storage under $\epsilon$-privacy}

\begin{document}

\author{Jayesh Seshadri, Prasad Krishnan

\thanks{Jayesh and Dr. Krishnan are with the Signal Processing and Communications Research Center, International Institute of Information Technology, Hyderabad, 500032, India (email: $\{$jayesh.seshadri@research., prasad.krishnan@$\}$iiit.ac.in). 
}}
\maketitle

\allowdisplaybreaks

\begin{abstract}
    The problem of Private Information Retrieval (PIR) across graph-based replicated systems has seen a surge of interest in recent literature. In such systems, servers are modeled as the vertices of a graph, and each file is represented by an edge connecting the pair of servers that store it. While previous research on Graph-PIR (G-PIR) has focused mostly on perfect information-theoretic privacy, relaxing this constraint has the potential for significantly higher retrieval rates. In this work, we consider the setting of Graph-based Weak Private Information Retrieval (G-WPIR), under an $\epsilon$-privacy metric. We adapt an existing G-WPIR scheme for arbitrary graphs to analyze its performance under the $\epsilon$-privacy metric. Additionally, we establish a lower bound on the expected download cost for arbitrary graphs under the $\epsilon$-privacy metric and derive from it an explicit rate upper bound for complete graphs. Finally, we show the achievable rate-privacy trade-offs for two specific graph classes—complete graphs and complete bipartite graphs.
\end{abstract}

\section{INTRODUCTION}
Private Information Retrieval (PIR) enables a client to fetch a specific file from a distributed multi-server storage system while ensuring that no individual server can identify which file is being requested. This privacy is achieved through a query-response mechanism known as a \textit{PIR protocol}. The performance of such a scheme is evaluated by its \textit{rate}, defined as the size of the requested file divided by the total number of bits downloaded from the servers. The \textit{capacity} of a given PIR system is the supremum of the rates over all PIR protocols. PIR has been the subject of extensive research, particularly concerning information-theoretically secure approaches. This surge in interest was catalyzed by the foundational studies of Sun and Jafar on replicated databases with both colluding and non-colluding servers \cite{sun2017capacity, sun2017_T_capacity}, and has since expanded to encompass coded storage models \cite{banawan2018capacity} as well as schemes optimizing the subpacketization \cite{tian2019capacity, Zhouetal_2019_TIT_MDS_optimalmsgsize_lowerupcost}.

Distributed storage systems can be naturally mapped to a (hyper)graph topology, where vertices denote servers and edges (or hyperedges) represent files replicated across two or multiple servers. Private Information Retrieval tailored for such architectures, known as Graph-PIR (G-PIR), was originally introduced in \cite{ravivetal_2018_TIT}. This work established the upper and lower capacity bounds for non-colluding and some colluding server models. Subsequently, G-PIR over simple graphs, where each file is assigned to exactly two distinct servers and no server pair shares more than one file, has seen significant advancements in \cite{Sadehetal_TIT_2023_BoundsonPIRgraphs,ge2025privateinformationretrievalgraphs,Kong_TIT_2025newcapacityboundspir,Yuhangetal_Capof4StarGraph_ISIT_2023}. These contributions specifically characterized the bounds on G-PIR capacity across several specialized network topologies, like star graphs, path graphs, complete graphs, and bipartite graphs. Crucially, because a retrieval protocol formulated for an $N$-vertex complete graph $K_N$ can be directly adapted for any arbitrary $N$-vertex graph, the lower bounds derived for $K_N$ serve as universal lower bounds for all graphs of the same size. Beyond simple graphs, the G-PIR framework has also been broadened to include multi-graph-based storage environments, where server pairs may share multiple files \cite{meel2025privateinformationretrievalmultigraphbased}, alongside developments in symmetric G-PIR formulations \cite{meel2025effectcommonrandomnessreplication,meel2025symmetric,meel2026role}.

A recent advancement in graph-based PIR includes a novel family of schemes applicable to both general and star-topology networks, as proposed in \cite{shanbhag2026private}. Specifically, the general graph protocol described in \cite[Section IV]{shanbhag2026private} leverages a randomized query generation mechanism that operates with the smallest possible \textit{subpacketization} level of $L=1$, which means that the files do not need to be partitioned into smaller fragments to execute the protocol. To initialize this approach, a pre-processing phase is required to partition the storage graph into disjoint independent sets. For an arbitrary storage graph $\mathcal{G}$ consisting of $N$ databases with an independence number of $\alpha(\mathcal{G})$, the scheme guarantees a retrieval rate lower bounded by $\max\{2/(2N-\alpha(\mathcal{G})), 1/(N-1)\}$. In the specific case of a complete graph $K_N$, this evaluates to exactly $1/(N-1)$. While this randomized approach currently outputs the highest known retrieval rates for balanced bipartite graphs, alternative protocols \cite{ge2025privateinformationretrievalgraphs} are known to achieve better rates for complete graphs.

Relaxing the strict privacy requirements of traditional PIR makes it possible to achieve higher retrieval rates at the expense of a controlled amount of information leakage to the servers. Protocols operating under this relaxed notion are known as Weak-PIR (WPIR) schemes. The formal study of WPIR has been established through several lines of research \cite{lin2021multi,jia2019capacity,samy2021asymmetric}. Although early WPIR models quantified this leakage on average—using metrics such as mutual information \cite{lin2021multi} or maximal leakage\cite{zhouetal_2020_WPIR_MaxL}, following the $\epsilon$-private retrieval of Toledo et al.~\cite{toledo2016epsilonldp}, Samy et al. \cite{samy2019capacity} introduced a leaky PIR definition centered on an $\epsilon$-privacy metric. Similar to $\epsilon$-local differential privacy ($\epsilon$-LDP) \cite{duchi2013ldpintro}, this metric strictly bounds the likelihood ratio between any possible realizations of the database queries using a non-negative constant $\epsilon$. Because the $\epsilon$-privacy constraint is enforced absolutely across all individual query realizations rather than merely in expectation, it provides substantially stronger and more rigorous privacy guarantees for each retrieved message. Pointwise leakage-ratio constraints have also been studied for replicated servers, including the scaling of the leakage-ratio exponent with the number of files \cite{zhao2025optimizing} and leaky PIR with side information \cite{huangfu2026leaky}; weakly-private retrieval from heterogeneously trusted servers is studied in \cite{huang2024weakly}, and the capacity of single-server weakly-private retrieval in \cite{lin2020capacity}. The aforementioned alternative relaxed privacy formulations based on Mutual Information leakage and Maximal leakage have subsequently been adapted across a wide range of WPIR architectures, including non-colluding replicated servers \cite{Anandetal_ISIT2025_WPIR, anand2026converseboundssunjafartypeweak}, MDS-coded storage systems \cite{orvedal2024weaklyprivateinformationretrievalmdscoded, Anandetal_ISIT2025_WPIR, anand2026converseboundssunjafartypeweak}, and colluding servers \cite{Anandetal_ISIT2025_WPIR, anand2026converseboundssunjafartypeweak}. Recently, the study of WPIR was extended to graph-based storage environments in \cite{Vidya_ISITA2026}, which evaluated a randomized protocol under average-case Mutual Information and Maximal Leakage metrics. In this work, we advance this framework by analyzing graph-based WPIR under the much stricter, $\epsilon$-privacy metric and establishing bounds on the achievable rate.

Privacy requirements for graph-based storage have also been relaxed in other directions very recently. In \emph{local PIR} \cite{nomeir2026local,meel2026local}, the desired index must be hidden only from the servers that store the desired file, and exact capacities are obtained for several graph classes (e.g., cycles); the framework of \cite{nomeir2026arbitrary} generalizes this to arbitrary per-server privacy requirements, interpolating between local and standard graph-based PIR. These works relax \emph{which} servers must be protected, whereas the $\epsilon$-privacy studied here bounds \emph{how much} every server may learn about $\theta$, yielding a smooth rate--privacy trade-off. The two directions are complementary.

In this work, we adapt the WPIR protocol introduced in \cite{Vidya_ISITA2026} to construct a retrieval scheme tailored for graph-based storage systems. To strictly measure and constrain the privacy leaked to the servers, we evaluate this scheme through the $\epsilon$-privacy framework for WPIR originally established in \cite{samy2019capacity,samy2021asymmetric}. The scheme itself is that of \cite{Vidya_ISITA2026}; what is new here is the pointwise $\epsilon$-privacy analysis of this scheme (in particular, its exact $\epsilon$ in terms of the bias $p$ and the upstream degrees) and the converse bound. The primary contributions of this work are as follows:
\begin{itemize}
    \item Section II contains the preliminaries, the system model and the G-WPIR principles evaluated through the lens of $\epsilon$-privacy.
    \item Section III details our proposed G-WPIR protocol, which yields adaptable tradeoffs between retrieval rate and privacy leakage. This approach serves as an $\epsilon$-private generalization of the graph-PIR protocol previously introduced in \cite[Section IV]{shanbhag2026private}, building on the scheme from \cite[Section III]{Vidya_ISITA2026}. We analyze the achievable rates of our scheme and its corresponding $\epsilon$-privacy guarantees.
    \item Section IV explores the bounds on G-WPIR subject to an $\epsilon$-privacy constraint, establishing a lower bound on the expected download cost for any graph. This yields an explicit upper bound on the achievable rate for complete graphs in terms of the number of servers (or nodes in the graph).
    \item Section V focuses on two specific graph classes: complete graphs and complete bipartite graphs. Utilizing our previous results of this work, we present their achievable rate-privacy tradeoffs together with the converse of Section IV for complete graphs.
\end{itemize}

In summary, this paper introduces the study of WPIR within graph-based storage environments under the pointwise $\epsilon$-privacy metric. This broadens the scope of existing literature that has previously examined WPIR in alternative distributed storage architectures.

 \textit{Notation:} For any positive integer $N$, let $[N]$ denote the set of integers $\{1,2,\dots,N\}$. For a collection of indexed random variables or sets $X$ and a subset of indices $\mathcal{I}\subseteq[N]$, we define $X_{\mathcal{I}} \triangleq \{X_i : i \in \mathcal{I}\}$. The symbols $\mathbf{0}$ and $\mathbf{1}$ represent the all-zero and all-one vectors, respectively, over the field $\mathbb{F}_{2}$, with lengths determined by context. The $\oplus$ symbol denotes the bit-wise logical XOR. A biased binary random variable is modeled as $X\sim \text{Bern}(p)$ where $\mathbb{P}(X=1)=1-\mathbb{P}(X=0)=p$. The special case $p=1/2$ reduces to a fair Bernoulli Trial. For a random variable $X$, $\mathbb{E}[X]$ denotes the expectation of $X$ taken with respect to its underlying probability measure. $H(X)$ and $H(X \mid Y)$ represent the Shannon entropy and conditional entropy, respectively, while $I(X;Y)$ denotes mutual information. An undirected graph $\mathcal{G}$ is defined as a pair $\mathcal{G}=(\mathcal{V},\mathcal{E})$, where $\mathcal{V} = [N]$ denotes the set of vertices (servers) and $\mathcal{E}$ denotes the set of edges (files), with the total number of files $K = |\mathcal{E}|$. An independent set of $\mathcal{G}$ is a subset of vertices in which no two vertices are adjacent, i.e., no edge in $\mathcal{E}$ connects any pair of vertices in the subset. The independence number, denoted by $\alpha(\mathcal{G})$, is defined as the cardinality of a largest independent set of $\mathcal{G}$.

\section{SYSTEM MODEL AND PRELIMINARIES}
\label{sec:sysmodel}

\subsection{Storage Model}
We consider a storage system consisting of $N$ servers, where each server is represented by a vertex of a simple, undirected graph $\mathcal{G}=(\mathcal{V},\mathcal{E})$. The vertex set is denoted by $\mathcal{V}=[N]\triangleq{\{1,2,\ldots,N\}}$ where each vertex corresponds to a distinct server in the system. The edge set $\mathcal{E}$ characterizes the file storage structure of the system, where each edge $\{i,j\}\in\mathcal{E}$ corresponds to a distinct file $W_{i,j}$ replicated across servers $i$ and $j$.
Accordingly, server $n\in\mathcal{V}$ stores all files associated with the edges incident to vertex $n$, namely:
\begin{equation*}
    \{W_{n,m}: \{n,m\}\in\mathcal{E}\}.
\end{equation*}
For convenience, these files are denoted by $W_n^1,W_n^2,\dots,W_n^{d_n}$, where $d_n=\left| \{m: \{n,m\}\in\mathcal{E}\} \right|$ is the degree of vertex $n$ in $\mathcal{G}$. Since each edge corresponds to one distinct file, the total number of files stored in the system is $K=|\mathcal{E}|$, and we also refer to the $K$ file indices as $[K]$.  The files are assumed to be mutually independent and uniformly distributed. Each file, $W_{i,j}$, has entropy $H(W_{i,j})=L\text{ bits}$, where $L\geq 1$ denotes the subpacketization level. Throughout this work, we focus on the unit-subpacketization setting for achievable schemes, in which each file consists of a single symbol, i.e., $L=1$. The converses hold for any $L$. 

\subsection{PIR Protocol}
A client wishes to privately retrieve a target file $W_{\theta}$, where the file index $\theta$ is selected uniformly at random from the edge set $\mathcal{E}$. The client communicates with the servers through a query-response protocol. Specifically, the client sends a query $Q_n$ to each server $n\in\mathcal{V}$ and receives a corresponding response $A_n$. The comprehensive set of all query realizations generated across the network is denoted by $\mathcal{Q} \triangleq \{Q_1, Q_2, \ldots, Q_N\}$. In a general G-WPIR scheme, the individual queries $Q_n:n\in[N]$ are generated as functions of the desired file index $\theta$ and private client-side randomness, and are statistically independent of the stored files. The response $A_n$ from server $n$ is a deterministic function of $Q_n$ and the files it contains.

For a PIR protocol to be correct, the client must be able to recover the desired file $W_\theta$ from the collection of queries and responses. This correctness requirement is expressed as
\begin{equation}
    H(W_{\theta}\mid \mathcal{Q},A_{[N]})=0, \label{eq:correctness}
\end{equation}
where $\mathcal{Q}$ and $A_{[N]}$ denote the collections of all queries and responses across the $N$ servers, respectively.

For a variable-download PIR protocol, the retrieval rate is defined as
\begin{equation}
    R\triangleq\frac{L}{\mathbb{E}[D]}, \label{eq:rate}
\end{equation}
where $L$ denotes the subpacketization level and $\mathbb{E}[D]$ represents the expected total download cost, given by:
\begin{equation}
    \mathbb{E}[D]=\sum_{n\in\mathcal{V}} H(A_n\mid Q_n). \label{eq:expecteddownload}
\end{equation}

The capacity of the graph-based PIR system, denoted by $C(\mathcal{G})$, is defined as the supremum of the retrieval rate over all correct PIR protocols (including non-linear protocols), i.e., $C(\mathcal{G})\triangleq\sup \{R\}$, where the supremum is taken over all PIR protocols satisfying the correctness condition.

\subsection{Strict and Weak Privacy}
The strict privacy requirement in PIR ensures that no individual server can obtain any information about the desired file index $\theta$ from the query-response pair it receives. i.e, 
\begin{equation}
    I(\theta ; Q_n, A_n) = 0 \quad \forall \, n \in \mathcal{V}. \label{eq:perfectprivacy}
\end{equation}
This requirement can also be represented as follows, letting $W_{\Omega}$ represent any arbitrary subset of the $K$ messages:
\begin{multline*}
    \left(A_n^{(k_1)}, Q_n^{(k_1)} \mid W_{\Omega}\right) \sim \left(A_n^{(k_2)}, Q_n^{(k_2)} \mid W_{\Omega}\right), \\ 
    \forall k_1,k_2 \in [K].
\end{multline*}

In this work, we consider a weak privacy requirement characterized by the privacy parameter $\epsilon$, which is non-negative. This metric was adopted by \cite{toledo2016epsilonldp} and was subsequently used in \cite{samy2019capacity}, referred to as \textit{$\epsilon$-privacy}.

To satisfy the $\epsilon$-privacy constraint, the ratio of the joint probability of queries and answers for any two different desired message indices must be bounded. Formally, for $\epsilon\geq 0$, a PIR protocol satisfies the $\epsilon$-privacy constraint if:
\begin{equation}
    \frac{\mathbb{P}\{A_n^{(\theta)} = \gamma, Q_n^{(\theta)} = \pi \mid  W_\Omega, \theta = k_1\}}{\mathbb{P}\{A_n^{(\theta)} = \gamma, Q_n^{(\theta)} = \pi \mid  W_\Omega, \theta = k_2\}}
    \le e^\epsilon \, \label{eq:epsprivacycomplete}
\end{equation}
for all servers $n \in [N]$ and for all distinct file indices $k_1, k_2 \in \{1:K\}$ where $k_1 \neq k_2$. Here, $\pi$ and $\gamma$ denote any possible realizations of the queries and answers, respectively. Since the response from any server $n$ is entirely determined by its received query and stored messages, the conditional probability of the answer realization remains independent of the target file index $\theta$. This gives us the following equation:
\begin{equation}
\begin{split}
    &\mathbb{P}\{A_n^{(\theta)} = \gamma \mid Q_n^{(\theta)} = \pi, W_\Omega, \theta = k_1\} \\
    &\qquad = \mathbb{P}\{A_n^{(\theta)} = \gamma \mid Q_n^{(\theta)} = \pi, W_\Omega, \theta = k_2\}. \label{eq:k1k2ansrealization}
\end{split}
\end{equation}
As a result, the privacy condition simplifies to bounding the likelihood ratio of the query probabilities:
\begin{equation}
    \frac{\mathbb{P}\{Q_n^{(\theta)} = \pi \mid W_\Omega, \theta = k_1\}}{\mathbb{P}\{Q_n^{(\theta)} = \pi \mid W_\Omega, \theta = k_2\}} \le e^\epsilon \, \label{eq:epsprivacy}
\end{equation}
or,
\begin{equation}
    \frac{\mathbb{P}\{Q_n^{(\theta)} = \pi \mid \theta = k_1\}}{\mathbb{P}\{Q_n^{(\theta)} = \pi \mid \theta = k_2\}} \le e^\epsilon. \, \label{eq:epsprivacyshort}
\end{equation}

\begin{remark}
    Note that $\epsilon=0$ is the perfect-privacy case. We note that an $\epsilon$-privacy guarantee implies corresponding bounds on both maximal leakage (i.e., $\mathrm{MaxL} \le \epsilon$) and mutual information leakage.
\end{remark}

\section{AN $\epsilon$-PRIVATE RETRIEVAL SCHEME FOR GENERAL GRAPHS}
\label{sec:epsprivscheme}

In this section, we present an independent-set-based $\epsilon$-private scheme for general graphs. This approach utilizes the linear PIR protocol architecture originally developed in \cite{shanbhag2026private} and subsequently adapted for average-leakage WPIR metrics in \cite{Vidya_ISITA2026}. While \cite{Vidya_ISITA2026} established the expected download cost as a function of a query bias $p$, our primary contribution here is the formulation and proof of the pointwise $\epsilon$-privacy guarantees for this scheme, mapping the bias parameter $p \in [1/2, 1]$ to a strict $\epsilon$-leakage constraint to parameterize the rate-privacy trade-off.

\subsection{Recalling the protocol from \cite{Vidya_ISITA2026}}
\label{subsec:protocolrecall}
We begin by recalling the construction of the protocol from \cite{Vidya_ISITA2026}. Let $\mathcal{G} = (\mathcal{V}, \mathcal{E})$ denote the storage graph comprising $N$ servers (vertices) and $K$ files (edges). To privately retrieve a target file $\theta \in \mathcal{E}$, the client first decomposes the vertex set into a \textit{sequential independent-set partition} $\mathcal{V} = \bigsqcup_{s=1}^{\kappa} \mathcal{I}_s$, where each $\mathcal{I}_s$ is iteratively chosen as a largest independent set of the remaining induced subgraph. For a server $n \in \mathcal{I}_s$, its degree $d_n$ splits into a downstream degree $d_n^\downarrow$  whose edges connect to servers in $\bigcup_{k=s+1}^{\kappa} \mathcal{I}_k$ and an upstream degree, denoted by $U_n = d_n^\uparrow = d_n - d_n^\downarrow$, with edges connecting to servers in $\bigcup_{k=1}^{s-1} \mathcal{I}_k$. Let $\mathcal{U}_n = \{e_1, e_2, \dots, e_{U_n}\} \subseteq \mathcal{E}$ denote the collection of these upstream edges for server $n$.

\noindent \textbf{Query Generation:} \\
For each server $n$, a binary query vector $Q_n \in \mathbb{F}_2^{d_n}$ is generated as follows:

\noindent \textit{Step 1: Queries for the initial set $\mathcal{I}_1$} \\
Servers in $\mathcal{I}_1$ have only downstream edges ($d_n^\uparrow = 0$). For each $n \in \mathcal{I}_1$, the client flips an independent biased coin $C_n$ such that $\mathbb{P}(C_n = 0) = p$ and $\mathbb{P}(C_n = 1) = 1 - p$. The query vector is given by:
\begin{equation}
Q_n =
\begin{cases}
\mathbf{1} \in \mathbb{F}_2^{d_n}, & \text{with probability } 1-p, \\
\mathbf{0} \in \mathbb{F}_2^{d_n}, & \text{with probability } p.
\end{cases} \label{eq:schemequeryvectorprob}
\end{equation}

\noindent \textit{Step $s$: Queries for subsequent sets $\mathcal{I}_s$ ($s \ge 2$)} \\
For any server $n \in \mathcal{I}_s$, the query $Q_n = (Q_n^\downarrow, Q_n^\uparrow)$ concatenates downstream and upstream subvectors:
\begin{itemize}
    \item \textit{Downstream subvector $Q_n^\downarrow$:} If downstream edges exist, a fresh coin $C_n \sim \text{Bern}(1-p)$ dictates whether $Q_n^\downarrow$ is filled with all $1$s or all $0$s. If no downstream edges exist, this subvector is dropped.
    \item \textit{Upstream subvector $Q_n^\uparrow$:} For the $j$-th upstream file $e_j$, let $X_j$ be the coin value chosen by the corresponding upstream server. The query bit is evaluated as $f_j^n = X_j \oplus \mathbb{I}[\theta = e_j]$. This simply mirrors the upstream server's query bit for non-target files, but flips it for the target file $\theta$.
\end{itemize}
The query then sent to server $n$ is represented by a binary vector
\begin{equation}
    Q_n=(f_1^n,f_2^n,\dots,f_{d_n}^n)\in\mathbb{F}_2^{d_n}, \label{eq:schemequeryvector}
\end{equation}
where $d_n$ denotes the number of files available at server $n$.

\noindent \textbf{Server Responses and Correctness:} \\
Upon receiving a non-zero query, server $n$ generates its response as the linear combination:
\begin{equation}
    A_n=\bigoplus_{j=1}^{d_n} f_j^n W_n^j, \label{eq:schemeans}
\end{equation}
where $\oplus$ denotes addition over the binary field $\mathbb{F}_2$, and $W_n^j$ represents the $j$-th file available at server $n$. A zero query, $Q_n=\mathbf{0}$, is interpreted as a null query, in which case the server does not transmit a response. The client then recovers the desired file by aggregating all received responses: $W_\theta = \bigoplus_{n \in \mathcal{V}} A_n$.

For any file $W_{m,n}$ shared by an upstream server $m$ and a downstream server $n$, the respective query bits $f_i^m$ and $f_j^n$ are perfectly matched if the index $\{m,n\} \neq \theta$. This structure allows all undesired files to cancel out in the summation of the received responses. For the desired file $W_\theta$, the query bit is logically inverted ($f_i^m = f_j^n \oplus 1$), ensuring that exactly one copy of $W_\theta$ survives the final sum $\bigoplus_n A_n$.

\subsection{$\epsilon$-Privacy Guarantees for Protocol in Subsection \ref{subsec:protocolrecall}}
The following theorem captures the rate-privacy pair achieved by the above-described scheme. The rate expression in Theorem \ref{thm:maintheoremepsprivate} follows from \cite[Theorem 1]{Vidya_ISITA2026}. The expression for $\epsilon$ is shown in Subsection \ref{subsec:proofofmaintheoremepsprivate}. 
\begin{theorem}
\label{thm:maintheoremepsprivate}
    For the sequential largest-independent-set decomposition $\mathcal{V} = \bigsqcup_{s=1}^{\kappa} \mathcal{I}_s$, the protocol in Section \ref{subsec:protocolrecall} is $\epsilon$-private with:
    \begin{equation}
    \label{eqn:epsvalue}
        \epsilon = \begin{cases}
            2 \ln\left(\frac{p}{1-p}\right), & \text{if } \max_n U_n \ge 2, \\
            \ln\left(\frac{p}{1-p}\right), & \text{if } \max_n U_n = 1.
        \end{cases}
    \end{equation}
    The protocol with parameter $p\in[1/2,1)$ achieves rate $R = 1/{\mathbb{E}[D(p)]}$, where the expected download cost is given by
    \begin{align}
        \mathbb{E}[D(p)]
        &= |\mathcal{I}_{1}|\cdot(1-p) \nonumber \\
        &\;\;+ \sum_{n\notin\mathcal I_1}
        \!\!\left[1 - \frac{p^{{U_n}+{\mathbb I(d_n^{\downarrow}>0)}}}{K}\!\left(K -
        \frac{U_n(2p-1)}{p}\right)\right].
        \label{eq:schemedownloadcost}
    \end{align}
\end{theorem}
In Corollary \ref{corr:ratelowerboundforgivenPeps} that follows, the lower bound on the rate is obtained from \cite[Corollary 1]{Vidya_ISITA2026}, while the value of $p_\epsilon$ holds from \eqref{eqn:epsvalue}. 
\begin{corollary}
\label{corr:ratelowerboundforgivenPeps}
    For any graph $\mathcal{G}$ with independence number $\alpha$, the protocol achieves an $\epsilon$-private retrieval rate bounded by:
    \begin{equation}
        R_\epsilon \ge \frac{1}{N - \alpha p_\epsilon},
    \end{equation}
    where $p_\epsilon = \frac{e^{\epsilon/2}}{1+e^{\epsilon/2}}$ when $\max_n U_n \ge 2$, and $p_\epsilon = \frac{e^{\epsilon}}{1+e^{\epsilon}}$ when $\max_n U_n = 1$.
\end{corollary}
\subsection{Proof of Theorem \ref{thm:maintheoremepsprivate}}
\label{subsec:proofofmaintheoremepsprivate}
The rate in Theorem \ref{thm:maintheoremepsprivate} is calculated by using the following expression.
\begin{align*}
    {\mathbb E}(D(p))=\sum_{n\in[N]}P(Q_n\neq 0). 
\end{align*}
Here, $P(Q_n\neq 0)$ is the probability that the query sent to the $n^{\text{th}}$ server is not the all-zero query (which results in the empty response). By carefully calculating these probabilities, the rate expression in Theorem \ref{thm:maintheoremepsprivate} is obtained (please see \cite[Section IV]{Vidya_ISITA2026} for the full calculations). 
 
To prove the privacy guarantees of the scheme, we establish the following claims.

\begin{claim}
\label{claim:expecteddownloadisnonincreasing}
    $\mathbb{E}[D(p)]$ in~\eqref{eq:schemedownloadcost} is non-increasing in $p\in[1/2,1)$. 
\end{claim}
\begin{claim}
\label{claim:isbepsilonprivacy}
    To ensure $\epsilon$-privacy for the scheme in Subsection \ref{subsec:protocolrecall}, it is sufficient that the bias probability $p \in [1/2, 1)$ satisfies the inequality:
\begin{equation}
    p \le 
    \begin{cases}
    \frac{e^{\epsilon/2}}{1 + e^{\epsilon/2}}, & \text{if } \max_{n} U_n \ge 2, \\
    \frac{e^\epsilon}{1 + e^\epsilon}, & \text{if } \max_{n} U_n = 1.
    \end{cases} \label{eq:epsprivacypval}
\end{equation}
\end{claim}
By Claims \ref{claim:expecteddownloadisnonincreasing} and \ref{claim:isbepsilonprivacy}, it follows that we must choose the largest $p$ that satisfies $\epsilon$-privacy, that is, the largest $p$ in \eqref{eq:epsprivacypval}, for getting the largest rate. The value of $\epsilon$ in \eqref{eqn:epsvalue} follows, and then Theorem \ref{thm:maintheoremepsprivate} is proved. The proofs for the two claims now follow.

\textit{\underline{Proof of Claim \ref{claim:expecteddownloadisnonincreasing}:}}
Assuming $\mathcal{G}$ has no isolated vertices, we can express the total download cost as $\mathbb{E}[D(p)] = N - \frac{1}{K}\sum_{\theta\in\mathcal{E}}S_\theta(p)$, where $S_\theta(p) \triangleq \sum_{v\in\mathcal{V}}P(Q_v=\mathbf{0}\mid\theta)$. We first argue that for any fixed target file $\theta$, $S_\theta(p)$ is non-decreasing in $p$ and then average this result over all $\theta$. 

Fix $\theta=\{m,n\}$ with endpoints in sets $m \in \mathcal{I}_t$ and $n \in \mathcal{I}_s$ ($t < s$). Let $\tau_v \triangleq \mathbb{I}(d_v^{\downarrow}>0)+U_v$ represent the effective number of unbiased constraints for server $v$. For all $v \neq n$, the query generation is independent of $\theta$, giving $P(Q_v=\mathbf{0}\mid\theta) = p^{\tau_v}$. For server $n$, exactly one upstream bit corresponding to $\theta$ is logically inverted, evaluating to zero with probability $(1-p)$, which gives $P(Q_n=\mathbf{0}\mid\theta) = p^{\tau_n-1}(1-p)$.

Differentiating $S_\theta(p)$ with respect to $p$ yields:
\begin{align}
    \frac{dS_\theta}{dp} &= \frac{d}{dp} \Bigg( \sum_{v\neq n} p^{\tau_v} + p^{\tau_n-1}(1-p) \Bigg) \nonumber \\
    &= \sum_{v\neq n}\tau_v\,p^{\tau_v-1} + (\tau_n-1)\,p^{\tau_n-2}(1-p) - p^{\tau_n-1}. \label{eqn:dsbydtheta}
\end{align}
Since $p \in [1/2, 1)$ and $\tau_n \ge 1$, all terms in \eqref{eqn:dsbydtheta} are non-negative except for $-p^{\tau_n-1}$. Now because the initial set $\mathcal{I}_1$ is a maximal independent set and $n \notin \mathcal{I}_1$, $n$ is guaranteed to have at least one neighbor $v_0 \in \mathcal{I}_1$. This neighbor possesses no upstream edges ($U_{v_0} = 0$) and at least one downstream edge ($d_{v_0}^\downarrow \ge 1$), yielding exactly $\tau_{v_0} = 1$. 

Thus, $v_0$'s contribution to the summation in \eqref{eqn:dsbydtheta} evaluates exactly to $\tau_{v_0}p^{\tau_{v_0}-1} = 1$. Dropping the remaining non-negative terms isolates the lower bound:
\begin{equation}
    \frac{dS_\theta}{dp} \ge 1 - p^{\tau_n-1} \ge 0.
\end{equation}
Thus $S_\theta(p)$ is non-decreasing in $p$ for every $\theta$. Averaging over $\theta$ shows that $\mathbb{E}[D(p)]$ is non-increasing in $p$, and hence the claim is proved.

\textit{\underline{Proof of Claim \ref{claim:isbepsilonprivacy}:}}
    For a server $n \in \mathcal{I}_s$ (where $s \ge 2$), the upstream bit $f_j^n$ corresponding to the shared file $W_{n,m}$ (denoted by index $e_j$) is constructed as $f_j^n = X_j \oplus \mathbb{I}(\theta = e_j)$. Here, $\mathbb{I}(\cdot)$ is the indicator function that outputs 1 if the condition is true and 0 otherwise.

We note the following useful observation. 
\begin{observation}
\label{OBS:queryprivacystructure}
For server $n$, let the query $Q_n = (Q_n^\uparrow,Q_n^\downarrow)$ be represented by its upstream and downstream subvectors respectively. WLOG assume that $q^{\uparrow} = (f_{1}^n,....,f_{U_n}^n)$ and $q^{\downarrow} = (f_{U_n+1}^n,....,f_{d_n}^n)$. The downstream subvector is determined entirely by a single coin $C_n$, meaning it is independent of $\theta$ and the upstream bits. The upstream subvector $Q_n^\uparrow = q^{\uparrow}$ has $U_n$ bits which are mutually independent, as each upstream bit is a function of a different coin originating at a different distinct server in a prior inependent set. We factor the conditional probability for any query realization as follows:
\begin{equation*}
    \mathbb{P}\{Q_n = q \mid \theta = k\} = \mathbb{P}\{Q_n^{\downarrow} = q^{\downarrow}\} \prod_{j=1}^{U_n} \mathbb{P}\{Q^{\uparrow}_{n,j} = f^{n}_j \mid \theta = k\}.
\end{equation*}
Now, considering the $\epsilon$-privacy constraint in \eqref{eq:epsprivacyshort}, we write
\begin{align}
    &\frac{\mathbb{P}\{Q_n = q \mid \theta = k_1\}}{\mathbb{P}\{Q_n = q \mid \theta = k_2\}} \nonumber \\
    &\quad = \frac{\mathbb{P}\{Q_n^{\downarrow} = q^{\downarrow}\} \prod_{j=1}^{U_n} \mathbb{P}\{Q^{\uparrow}_{n,j} = f^{n}_j \mid \theta = k_1\}}{\mathbb{P}\{Q_n^{\downarrow} = q^{\downarrow}\} \prod_{j=1}^{U_n} \mathbb{P}\{Q^{\uparrow}_{n,j} = f^{n}_j \mid \theta = k_2\}} \nonumber \\
    &\quad = \frac{\prod_{j=1}^{U_n} \mathbb{P}\{Q^{\uparrow}_{n,j} = f^{n}_j \mid \theta = k_1\}}{\prod_{j=1}^{U_n} \mathbb{P}\{Q^{\uparrow}_{n,j} = f^{n}_j \mid \theta = k_2\}}. \label{eqn:queryproductstructure}
\end{align}
We thus note that, in the ratio in \eqref{eqn:queryproductstructure}, those terms in the numerator and the denominator which are not dependent on the choices $\theta=k_1$ and $\theta=k_2$ cancel out each other. Specifically, this cancellation happens for all downstream bits of $Q_n$. Further, such cancellation also happens for each upstream bit that is related to an upstream edge which is neither $k_1$ nor $k_2$.

\end{observation}

Clearly, by Observation \ref{OBS:queryprivacystructure}, we see that the ratio in \eqref{eqn:queryproductstructure} is $1$, if none of the upstream edges at server $n$ corresponds to the files $k_1$ or $k_2$. In this case, the privacy condition given by \eqref{eqn:queryproductstructure} is satisfied with $\epsilon=0$.  

Now, we consider the other two cases, i.e., when at least one of the upstream edges at server $n$ is $k_1$ or $k_2$. 
\begin{itemize}
    \item \textit{Case 1: Exactly one of the  files, say file $k_1$, corresponds to an upstream edge $e_j$ for server $n$.} \\
    In this scenario, the query probability differs in exactly one bit, $f_j^n$. Assume we observe $f_j^n = 1$: \\
    If $\theta=k_1$ (i.e., desired file corresponds to the upstream edge $e_j$) then,
    \begin{align*}
        \mathbb{P}\{f_j^n = 1 \mid \theta = k_1\} &= \mathbb{P}\{X_j \oplus 1 = 1 \mid \theta = k_1\} \\
        &= \mathbb{P}\{X_j = 0\} = p.
    \end{align*}
    If the desired file is $k_2$ (not an upstream edge for server $n$) then,
    \begin{align*}
        \mathbb{P}\{f_j^n = 1 \mid \theta = k_2\} &= \mathbb{P}\{X_j \oplus 0 = 1 \mid \theta = k_2\} \\
        &= \mathbb{P}\{X_j = 1\} = 1 - p.
    \end{align*}
    The resulting ratio of the conditional query probabilities, by Observation \ref{OBS:queryprivacystructure} and \eqref{eqn:queryproductstructure}, is thus obtained as 
    $$\frac{\prod_{j=1}^{U_n} \mathbb{P}\{Q^{\uparrow}_{n,j} = f^{n}_j \mid \theta = k_1\}}{\prod_{j=1}^{U_n} \mathbb{P}\{Q^{\uparrow}_{n,j} = f^{n}_j \mid \theta = k_2\}}=\frac{p}{1-p}.$$
    When the observed bit is $f_j^n = 0$, we can compute the same ratio as $\frac{1-p}{p}$, which is bounded by $\frac{p}{1-p}$ since $p \ge 1/2$.

    \item \textit{Case 2: Both $k_1$ and $k_2$ are distinct upstream edges of server $n$.} \\
    This case is only possible if server $n$ has $U_n \ge 2$ upstream edges. Here, the query probabilities differ across two bits simultaneously, say $f_j^n$ and $f_{j'}^n$, driven by independent upstream coins $X_j$ and $X_{j'}$. Assume we observe a query realization where $f_j^n = 1$ and $f_{j'}^n = 0$: \\
    If the desired file is $k_1$ (corresponding to $e_j$) then,
    \begin{align*}
        &\mathbb{P}\{f_j^n = 1, f_{j'}^n = 0 \mid \theta = k_1\} \\
        &\quad = \mathbb{P}\{X_j \oplus 1 = 1, X_{j'} \oplus 0 = 0 \mid \theta = k_1\} \\
        &\quad = \mathbb{P}\{X_j = 0\} \mathbb{P}\{X_{j'} = 0\} = p^2.
    \end{align*}
    If the desired file is $k_2$ (corresponding to $e_{j'}$) then,
    \begin{align*}
        &\mathbb{P}\{f_j^n = 1, f_{j'}^n = 0 \mid \theta = k_2\} \\
        &\quad = \mathbb{P}\{X_j \oplus 0 = 1, X_{j'} \oplus 1 = 0 \mid \theta = k_2\} \\
        &\quad = \mathbb{P}\{X_j = 1\} \mathbb{P}\{X_{j'} = 1\} = (1 - p)^2.
    \end{align*}
    The resulting probability ratio, again by Observation \ref{OBS:queryprivacystructure} and \eqref{eqn:queryproductstructure}, is obtained as
\begin{align*}
    &\frac{\prod_{j=1}^{U_n} \mathbb{P}\{Q^{\uparrow}_{n,j} = f^{n}_j \mid \theta = k_1\}}{\prod_{j=1}^{U_n} \mathbb{P}\{Q^{\uparrow}_{n,j} = f^{n}_j \mid \theta = k_2\}}\\
    &\quad=\frac{p^2}{(1-p)^2} = \left(\frac{p}{1-p}\right)^2.
 \end{align*}
\end{itemize}
Now, for any other realization of $(f_1^n,\hdots,f_{d_n}^n)$ as per the protocol (i.e., for $(f_j^n,f_{j'}^n)=(0,0),(1,1)$ and $(0,1)$, we get the ratios of the conditional query probabilities in \eqref{eqn:queryproductstructure} as $p(1-p)/((1-p)p)$ or $((1-p)/p)^2$, both of which are upper bounded by $(p/(1-p))^2$ as $p\geq 1/2$. Therefore, for Case 2, it is sufficient to choose $p$ so as to bound the ratio $(p/(1-p))^2$ by $e^\epsilon.$

Thus, we have the following conditions for the choice $p$ in the protocol, to meet the $\epsilon$-privacy constraint. 

If $\max_{n} U_n = 1$, Case 2 cannot occur. We require:
\begin{equation*}
    \frac{p}{1-p} \le e^\epsilon \implies p \le \frac{e^\epsilon}{1 + e^\epsilon}.
\end{equation*}

If $\max_{n} U_n \ge 2$, Case 2 defines the maximum leakage. We thus require:
\begin{equation*}
    \left(\frac{p}{1-p}\right)^2 \le e^\epsilon \implies \frac{p}{1-p} \le e^{\epsilon/2} \implies p \le \frac{e^{\epsilon/2}}{1 + e^{\epsilon/2}}.
\end{equation*}
This concludes the proof of Claim \ref{claim:isbepsilonprivacy} and hence of Theorem \ref{thm:maintheoremepsprivate}.

\section{Converse Bound for PIR Rate under $\epsilon$-privacy}
In this section, we establish theoretical bounds on the achievable retrieval rate of G-WPIR systems under the $\epsilon$-privacy constraint. Specifically, we adapt the fundamental entropy inequalities originally developed for perfectly private graph-based PIR by Sadeh et al. \cite{Sadehetal_TIT_2023_BoundsonPIRgraphs} to accommodate $\epsilon$-information leakage. We first establish foundational inequalities, which we then apply to derive a lower bound on the total download cost for arbitrary graphs, followed by a specialized rate bound for complete graphs.

\subsection{Essential Inequalities under $\epsilon$-Privacy}

In this subsection, we prove a sequence of essential lemmas that establish a lower bound on the total download, and subsequently the rate, in any $\epsilon$-private protocol. Lemma \ref{lemma:epsilon_privacy_k1k2_entropy} bounds the ratio of conditional entropies across different target file indices, while Lemma \ref{lemma:epsilon_privacy_theta_entropy} bounds the overall conditional entropy of the answer given the query and a message subset. We later combine these arguments to prove the lower bound on the total download cost in any $\epsilon$-private protocol in Theorem \ref{thm:rate_general_graphs}.

\begin{lemma}
\label{lemma:epsilon_privacy_k1k2_entropy}
    For any server $n$, subset of messages $W_\Omega \subseteq W_{[K]}$, and any two distinct file indices $k_1, k_2 \in [K]$, the conditional entropy of the answer -- given all the query realizations and the message subset and the desired file index -- satisfies the following upper and lower bounds:
\begin{align}
    &e^{-\epsilon} H(A_n \mid \mathcal{Q}, W_\Omega, \theta = k_2) \le H(A_n \mid \mathcal{Q}, W_\Omega, \theta = k_1) \nonumber\\
    &\qquad\qquad\qquad\le e^{\epsilon} H(A_n \mid \mathcal{Q}, W_\Omega, \theta = k_2). \label{eq:k1k2entropybound}
\end{align}
\end{lemma}
\begin{IEEEproof}
    First, note that the answer $A_n$ is a deterministic function of the query $Q_n$ and the files stored at server $n$. Since the queries $\mathcal{Q} \setminus Q_{n}$ are independent of the stored files given $\theta$, we can reduce the conditioning on the joint queries $\mathcal{Q}$ to the marginal query $Q_n$:
\begin{equation}
    H(A_n \mid \mathcal{Q}, W_\Omega, \theta = k) = H(A_n \mid Q_n, W_\Omega, \theta = k). \label{eq:lem:k1k2entropy_0}
\end{equation}
    We expand the conditional entropy $H(A_n\mid Q_n, W_\Omega, \theta = k_1)$ over the individual query realization $q_n$ and the message subset $\underline{w}$ to establish the lower bound:
\begin{align}
    &H(A_n \mid Q_n, W_\Omega, \theta = k_1) \nonumber \\
    &\quad = \sum_{q_n, \underline{w}} P_{Q_n, W_\Omega \mid \theta}(q_n, \underline{w} \mid k_1) \nonumber \\
    &\qquad \times H(A_n \mid Q_n = q_n, W_\Omega = \underline{w}, \theta = k_1) \nonumber \\
    &\quad = \sum_{q_n, \underline{w}} P_{Q_n \mid \theta}(q_n \mid k_1) P_{W_\Omega}(\underline{w}) \nonumber \\
    &\qquad \times H(A_n \mid Q_n = q_n, W_\Omega = \underline{w}, \theta = k_1) \label{eq:lem:k1k2entropy_1} \\
    &\quad \ge e^{-\epsilon} \sum_{q_n, \underline{w}} P_{Q_n \mid \theta}(q_n \mid k_2) P_{W_\Omega}(\underline{w}) \nonumber \\
    &\qquad \times H(A_n \mid Q_n = q_n, W_\Omega = \underline{w}, \theta = k_2) \label{eq:lem:k1k2entropy_2} \\
    &\quad = e^{-\epsilon} H(A_n \mid Q_n, W_\Omega, \theta = k_2), \nonumber
\end{align}

where \eqref{eq:lem:k1k2entropy_1} follows because the query $Q_n$ and the stored messages $W_\Omega$ are generated independently, and the stored messages are independent of the target file index $\theta$; and \eqref{eq:lem:k1k2entropy_2} is obtained by substituting the $\epsilon$-privacy definition's lower bound for the marginal query, $P_{Q_n \mid \theta}(q_n \mid k_1) \ge e^{-\epsilon} P_{Q_n \mid \theta}(q_n \mid k_2)$, into \eqref{eq:lem:k1k2entropy_1} alongside the property that $H(A_n \mid Q_n = q_n, W_\Omega = \underline{w}, \theta = k_1) = H(A_n \mid Q_n = q_n, W_\Omega = \underline{w}, \theta = k_2)$ because the entropy of the answer does not change with the desired file index once the specific query and message contents are fixed. Substituting \eqref{eq:lem:k1k2entropy_0} back into both sides yields the stated lower bound.

The corresponding upper bound $H(A_n \mid \mathcal{Q}, W_\Omega, \theta = k_1) \le e^{\epsilon} H(A_n \mid \mathcal{Q}, W_\Omega, \theta = k_2)$ is obtained symmetrically by substituting the upper bound of the marginal $\epsilon$-privacy constraint, $P_{Q_n \mid \theta}(q_n \mid k_1) \le e^{\epsilon} P_{Q_n \mid \theta}(q_n \mid k_2)$, into \eqref{eq:lem:k1k2entropy_1} and following identical algebraic steps.\end{IEEEproof}

\begin{lemma}
\label{lemma:epsilon_privacy_theta_entropy}
    For any server $n$, the overall conditional entropy of the answer given the query and a message subset satisfies the following upper and lower bounds:
{\small \begin{align}
    &\left( \frac{K}{1 + (K - 1)e^{\epsilon}} \right) H(A_n \mid \mathcal{Q}, W_\Omega) \le H(A_n \mid \mathcal{Q}, W_\Omega, \theta = k_1) \nonumber\\
    &\qquad\qquad\le \left( \frac{K}{1 + (K - 1)e^{-\epsilon}} \right) H(A_n \mid \mathcal{Q}, W_\Omega).
    \label{eqn:boundonanswerentropygivenk1}
\end{align}
}
\end{lemma}
\begin{IEEEproof} We first establish the upper bound. Consider the following difference between the entropies:
\begin{flalign}
    &H(A_n \mid \mathcal{Q}, W_\Omega) - H(A_n \mid \mathcal{Q}, W_\Omega, \theta = k_1) \nonumber \\
    &\qquad = H(A_n \mid \mathcal{Q}, W_\Omega, \theta) - H(A_n \mid \mathcal{Q}, W_\Omega, \theta = k_1) \label{eq:lem:thetaentropy_1} \\
    &\qquad = \sum_{k} P_\theta(k) H(A_n \mid \mathcal{Q}, W_\Omega, \theta = k) \nonumber \\
    &\qquad \qquad \qquad - \sum_{k} P_\theta(k) H(A_n \mid \mathcal{Q}, W_\Omega, \theta = k_1) \nonumber \\
    &\qquad = \sum_{k \neq k_1} P_\theta(k) \Big[ H(A_n \mid \mathcal{Q}, W_\Omega, \theta = k) \nonumber \\
    &\qquad \qquad \qquad - H(A_n \mid \mathcal{Q}, W_\Omega, \theta = k_1) \Big] \nonumber \\
    &\qquad \ge \sum_{k \neq k_1} P_\theta(k) \Big[ e^{-\epsilon} H(A_n \mid \mathcal{Q}, W_\Omega, \theta = k_1) \nonumber \\
    &\qquad \qquad \qquad - H(A_n \mid \mathcal{Q}, W_\Omega, \theta = k_1) \Big] \label{eq:lem:thetaentropy_2} \\
    &\qquad = \left( 1 - \frac{1}{K} \right) (e^{-\epsilon} - 1) H(A_n \mid \mathcal{Q}, W_\Omega, \theta = k_1), \label{eq:lem:thetaentropy_3} 
\end{flalign}
 where \eqref{eq:lem:thetaentropy_1} follows because $A_n$ is independent of $\theta$ given the queries and $W_\Omega$ i.e., $H(A_n \mid \mathcal{Q}, W_\Omega) = H(A_n \mid \mathcal{Q}, W_\Omega, \theta)$; \eqref{eq:lem:thetaentropy_2} follows from Lemma \ref{lemma:epsilon_privacy_k1k2_entropy}; \eqref{eq:lem:thetaentropy_3} evaluates the summation by assuming a uniform prior distribution for the desired file, $P_\theta(k) = \frac{1}{K}$, over the remaining $K-1$ terms. Rearranging the terms gives us the upper bound in \eqref{eqn:boundonanswerentropygivenk1}.
 We obtain the corresponding lower bound as follows:
 \begin{align}
    &H(A_n \mid \mathcal{Q}, W_\Omega) - H(A_n \mid \mathcal{Q}, W_\Omega, \theta = k_1) \nonumber  \\
     &\qquad = \sum_{k \neq k_1} P_\theta(k) \Big[ H(A_n \mid \mathcal{Q}, W_\Omega, \theta = k) \nonumber \\
    &\qquad \qquad \qquad - H(A_n \mid \mathcal{Q}, W_\Omega, \theta = k_1) \Big] \label{eq:lem:thetaentropy_4} \\
    &\qquad \le \sum_{k \neq k_1} P_\theta(k) \Big[ e^{\epsilon} H(A_n \mid \mathcal{Q}, W_\Omega, \theta = k_1) \nonumber \\
    &\qquad \qquad \qquad - H(A_n \mid \mathcal{Q}, W_\Omega, \theta = k_1) \Big] \label{eq:lem:thetaentropy_5} \\
    &\qquad = \left( 1 - \frac{1}{K} \right) (e^{\epsilon} - 1) H(A_n \mid \mathcal{Q}, W_\Omega, \theta = k_1), \label{eq:lem:thetaentropy_6}
 \end{align}
 where \eqref{eq:lem:thetaentropy_4} is obtained by following the same steps as the proof for the upper bound; \eqref{eq:lem:thetaentropy_5} follows from the upper bound established in Lemma \ref{lemma:epsilon_privacy_k1k2_entropy}. Now, rearranging the terms in \eqref{eq:lem:thetaentropy_6} gives us lower bound in \eqref{eqn:boundonanswerentropygivenk1}.
\end{IEEEproof}

 Next, we bound the entropy of $W_{i,j}$ by the sum of the conditional entropies of the answers given the queries and the remaining files of servers $i$ and $j$. The property shown in the following lemma is used to establish the rate upper bound for complete graphs.
\begin{lemma}
\label{lemma:subpack_entropy_bound}
    Let $\mathcal{W}$ denote the set of all files. Let $i,j \in [N]$ be two distinct servers that share the file $W_{i,j} \in \mathcal{W}$, then
\begin{equation}
\begin{split}
    L &\le \frac{K}{1+(K-1)e^{-\epsilon}} \Big[ H(A_{i} \mid \mathcal{W} \setminus \{W_{i,j}\}, \mathcal{Q}) \\
    &\quad + H(A_{j} \mid \mathcal{W} \setminus \{W_{i,j}\}, \mathcal{Q}) \Big].
\end{split}
\end{equation}
\end{lemma}
\begin{IEEEproof}
    The proof proceeds in a similar manner to that of \cite[Lemma 3]{Sadehetal_TIT_2023_BoundsonPIRgraphs}. Let $k$ be the index corresponding to the file $W_{i,j}$.
\begin{align}
    L &= H(W_{i,j}) \nonumber \\
    &= H(W_{i,j} \mid \mathcal{W} \setminus \{W_{i,j}\}, \mathcal{Q}, \theta=k) \label{eq:lem:subpackentropy_1} \\
    &= H(W_{i,j} \mid \mathcal{W} \setminus \{W_{i,j}\}, \mathcal{Q}, \theta=k) \nonumber \\
    &\qquad- H(W_{i,j} \mid \mathcal{W} \setminus \{W_{i,j}\}, \mathcal{Q}, \theta=k, A_{[N]}) \label{eq:lem:subpackentropy_2} \\
    &= H(W_{i,j} \mid \mathcal{W} \setminus \{W_{i,j}\}, \mathcal{Q}, \theta=k) \nonumber \\
    &\qquad- H(W_{i,j} \mid \mathcal{W} \setminus \{W_{i,j}\}, \mathcal{Q}, \theta=k, A_{i}, A_{j}) \label{eq:lem:subpackentropy_3} \\
    &= I(W_{i,j}; A_{i}, A_{j} \mid \mathcal{W} \setminus \{W_{i,j}\}, \mathcal{Q}, \theta=k) \nonumber \\
    &= H(A_{i}, A_{j} \mid \mathcal{W} \setminus \{W_{i,j}\}, \mathcal{Q}, \theta=k) \nonumber \\
    &\le H(A_{i} \mid \mathcal{W} \setminus \{W_{i,j}\}, \mathcal{Q}, \theta=k) \nonumber \\
    &\qquad+ H(A_{j} \mid \mathcal{W} \setminus \{W_{i,j}\}, \mathcal{Q}, \theta=k) \nonumber \\
    &\le \frac{K}{1+(K-1)e^{-\epsilon}} \Big[ H(A_{i} \mid \mathcal{W} \setminus \{W_{i,j}\}, \mathcal{Q}) \nonumber \\
    &\qquad + H(A_{j} \mid \mathcal{W} \setminus \{W_{i,j}\}, \mathcal{Q}) \Big], \label{eq:lem:subpackentropy_4}
\end{align}
where \eqref{eq:lem:subpackentropy_1} follows since $W_{i,j}$ is independent of the remaining files, the queries and the requested file $\theta = k$; \eqref{eq:lem:subpackentropy_2} follows from the correctness condition; \eqref{eq:lem:subpackentropy_3} follows since answers $A_{[N]} \setminus \{A_i, A_j\}$ are a deterministic function of $\mathcal{Q}$ and $\mathcal{W} \setminus \{W_{i,j}\}$; and \eqref{eq:lem:subpackentropy_4} follows from the upper bound of Lemma \ref{lemma:epsilon_privacy_theta_entropy} by substituting $W_\Omega = \mathcal{W} \setminus \{W_{i,j}\}$.
\end{IEEEproof}

The following theorem establishes an essential bound on the total download cost for any $\epsilon$-private protocol over arbitrary graphs.
\begin{theorem}
    For any $\epsilon$-private PIR scheme over an arbitrary graph G, the following inequality holds:
\label{thm:rate_general_graphs}
\begin{equation}
\begin{split}
    \sum_{j=1}^{N}H(A_{j}|\mathcal{Q}) &\ge \frac{1+(K-1)e^{-\epsilon}}{K} \Bigg[ L + \frac{K}{1+(K-1)e^{\epsilon}} \\
    &\times \Big( H(A_{1}|\mathcal{Q},W_{k}) + \sum_{j=2}^{N}H(A_{j}|W_{[j-1]},\mathcal{Q}) \Big) \Bigg],
\end{split}
\end{equation}
where $W_k$ is a file stored at server 1 and $W_{[j-1]}$ denotes the set of all files stored at servers 1 through $j-1$.
\end{theorem}
\begin{IEEEproof} We begin the proof by following a similar structure to \cite[Lemma 2.5]{ge2025privateinformationretrievalgraphs}. Let $W_{\{i\}}$ indicate the set of files stored in server $i$. Then, $W_k \in W_{\{1\}} \subseteq W_{[j-1]}$, for all $j \ge 2$. Now,
\begin{align}
    L&= H(W_{k}) = H(W_{k}|\mathcal{Q},\theta=k) \label{eq:thm:rategeneral_1} & \\
    &= H(W_{k}|\mathcal{Q},\theta=k) - H(W_{k}|A_{[N]},\mathcal{Q},\theta=k) \label{eq:thm:rategeneral_2} & \\
    &= I(W_{k};A_{[N]}|\mathcal{Q},\theta=k) \nonumber & \\
    &= H(A_{[N]}|\mathcal{Q},\theta=k) - H(A_{[N]}|\mathcal{Q},W_{k},\theta=k) \nonumber & \\
    &\le \sum_{j=1}^{N}H(A_{j}|\mathcal{Q},\theta=k) - H(A_{[N]}|\mathcal{Q},W_{k},\theta=k) \nonumber & \\
    &\le \frac{K}{1+(K-1)e^{-\epsilon}}\sum_{j=1}^{N}H(A_{j}|\mathcal{Q}) \nonumber & \\
    &\qquad \qquad- H(A_{[N]}|\mathcal{Q},W_{k},\theta=k). \label{eq:thm:rategeneral_3}
\end{align}
Rearranging the terms gives us:
\begin{align}
   \sum_{j=1}^{N}&H(A_{j}|\mathcal{Q})\\
   &\ge \frac{1+(K-1)e^{-\epsilon}}{K}\Big[L+H(A_{[N]}|\mathcal{Q},W_{k},\theta=k)\Big] \nonumber & \\
    &= \frac{1+(K-1)e^{-\epsilon}}{K} \nonumber & \\
    &\qquad \qquad\times \Big[L+\sum_{j=1}^{N}H(A_{j}|A_{[j-1]},\mathcal{Q},W_{k},\theta=k)\Big] \label{eq:thm:rategeneral_4} & \\
    &= \frac{1+(K-1)e^{-\epsilon}}{K}\Big[L+H(A_{1}|\mathcal{Q},W_{k},\theta=k) \nonumber & \\
    &\hspace{15pt} +\sum_{j=2}^{N}H(A_{j}|A_{[j-1]},\mathcal{Q},W_{k},\theta=k)\Big] \nonumber & \\
    &\ge \frac{1+(K-1)e^{-\epsilon}}{K}\Big[L+H(A_{1}|\mathcal{Q},W_{k},\theta=k) \nonumber & \\
    &\hspace{15pt} +\sum_{j=2}^{N}H(A_{j}|A_{[j-1]},W_{[j-1]},\mathcal{Q},W_{k},\theta=k)\Big] \label{eq:thm:rategeneral_5} & \\
    &= \frac{1+(K-1)e^{-\epsilon}}{K}\Big[L+H(A_{1}|\mathcal{Q},W_{k},\theta=k) \nonumber & \\
    &\hspace{15pt} +\sum_{j=2}^{N}H(A_{j}|W_{[j-1]},\mathcal{Q},\theta=k)\Big] \label{eq:thm:rategeneral_6} & \\
    &\ge \frac{1+(K-1)e^{-\epsilon}}{K}\Bigg[L + \frac{K}{1+(K-1)e^{\epsilon}} \nonumber \\
    &\quad \times \Big( H(A_{1}|\mathcal{Q},W_{k})+\sum_{j=2}^{N}H(A_{j}|W_{[j-1]},\mathcal{Q}) \Big) \Bigg], \label{eq:thm:rategeneral_7} &
\end{align}

where \eqref{eq:thm:rategeneral_1} follows since the queries $\mathcal{Q}$ and the identity of the desired file are independent of the specific file content; \eqref{eq:thm:rategeneral_2} follows from the correctness condition; \eqref{eq:thm:rategeneral_3} follows from the upper bound of Lemma \ref{lemma:epsilon_privacy_theta_entropy}; \eqref{eq:thm:rategeneral_4} follows from the chain rule of entropy; \eqref{eq:thm:rategeneral_5} follows since conditioning reduces entropy; \eqref{eq:thm:rategeneral_6} holds because $A_{[j-1]}$ can be completely determined by $W_{[j-1]}$ and $\mathcal{Q}$, and because $W_k \in W_{[j-1]}$ for any $j \ge 2$; and \eqref{eq:thm:rategeneral_7} follows from applying the lower bound of Lemma \ref{lemma:epsilon_privacy_theta_entropy} for the message subsets $W_k$ and $W_{[j-1]}$ to remove the conditioning on $\theta$. This finishes the proof.
\end{IEEEproof}

\subsection{$\epsilon$-Privacy Converse for Complete Graphs} We now focus on the specific class of complete graphs $K_N$. By leveraging the automorphism group of $K_N$, we formalize a symmetry condition in the following lemma. This lemma serves as the foundation for deriving an explicit lower bound on the rate of $\epsilon$-private schemes for complete graphs in the subsequent theorem.
\begin{lemma}
    Let $\Pi$ be an $\epsilon$-private scheme with for $K_N$ with rate $R$. Then there exists an $\epsilon$-private scheme $\Pi'$ with rate $R$ such that for all $i\neq j\in [N]$, the following conditions are true:
\label{lemma:transitive_graphs}
\begin{equation}
    H(A_i \mid \mathcal{Q}) = H(A_j \mid \mathcal{Q}), \label{eq:transitivegraphsequalentropy}
\end{equation}
and
\begin{equation}
    H(A_i \mid \mathcal{W} \setminus \{W_{i,j}\}, \mathcal{Q})\geqslant \frac{L(1+(K-1)e^{-\epsilon})}{2K}. \label{eq:completegraphsubpackentropy}
\end{equation}
\end{lemma}
\begin{IEEEproof}
    For an arbitrary PIR scheme $\Pi$ over $K_{N}$ achieving rate $R$, and an automorphism $\sigma \in \text{Aut}(K_{N})$ mapping edge $e$ (file index $\theta$) to $\sigma(e)$ (file index $\sigma(\theta)$), we construct a relabeled PIR scheme $\Pi_{\sigma}$. For each server $i$, the query (with its coordinates permuted according to the edge mapping induced by $\sigma$) and response are given by:
\begin{align}
    \mathcal{Q}_{i}^{\Pi_{\sigma}} &:= \mathcal{Q}_{\sigma(i)}^{\Pi}(\sigma(\theta)) .\label{eq:lem:automorphquery}\\
    A_{i}^{\Pi_{\sigma}} &:= A_{\sigma(i)}^{\Pi}. \label{eq:lem:automorphans}
\end{align}
Since $\sigma(K_{N})$ is a vertex permutation, $\Pi_{\sigma}$ maintains the same rate and correctness as $\Pi$. To satisfy structural symmetry conditions, we define a new scheme $\Pi'$ that uniformly samples and executes one of the schemes $\Pi_{\sigma}$ for $\sigma \in \text{Aut}(K_{N})$.
\begin{itemize}
    \item The user selects a desired file $\theta\in E(K_{N})$ and independently samples an automorphism $\sigma\in \text{Aut}(K_{N})$ uniformly at random, which is revealed to all servers.
    \item The user generates the queries $\mathcal{Q}^{\Pi'} := \mathcal{Q}^{\Pi_{\sigma}}$ and sends them to the corresponding servers.
    \item Each server $i$ responds with the answer $A_{i}^{\Pi'} := A_{i}^{\Pi_{\sigma}}$, which is received by the client.
\end{itemize}

Now, for any two distinct servers $i, j \in [N]$, there exists a automorphism $\tau \in \text{Aut}(K_{N})$ that swaps $i$ and $j$ while fixing all other vertices. We have:
\begin{align}
    &H(A_{i}^{\Pi'} \mid \mathcal{Q}^{\Pi'}) \nonumber \\
    &\quad= \frac{1}{|\text{Aut}(K_{N})|}\sum_{\sigma\in \text{Aut}(K_{N})} H(A_{i}^{\Pi_{\sigma}} \mid \mathcal{Q}^{\Pi_{\sigma}}) \label{eq:lem:transitivegraphs_1} \\
    &\quad= \frac{1}{|\text{Aut}(K_{N})|}\sum_{\sigma\in \text{Aut}(K_{N})} H(A_{\sigma(i)}^{\Pi} \mid \mathcal{Q}^{\Pi}) \nonumber \\
    &\quad= \frac{1}{|\text{Aut}(K_{N})|}\sum_{\sigma\in \text{Aut}(K_{N})} H(A_{\sigma(\tau(i))}^{\Pi} \mid \mathcal{Q}^{\Pi}) \label{eq:lem:transitivegraphs_2} \\
    &\quad= \frac{1}{|\text{Aut}(K_{N})|}\sum_{\sigma\in \text{Aut}(K_{N})} H(A_{\sigma(j)}^{\Pi} \mid \mathcal{Q}^{\Pi}) \label{eq:lem:transitivegraphs_3} \\
    &\quad= \frac{1}{|\text{Aut}(K_{N})|}\sum_{\sigma\in \text{Aut}(K_{N})} H(A_{j}^{\Pi_{\sigma}} \mid \mathcal{Q}^{\Pi_{\sigma}}) \nonumber \\
    &\quad= H(A_{j}^{\Pi'} \mid \mathcal{Q}^{\Pi'}). \nonumber
\end{align}
Here, \eqref{eq:lem:transitivegraphs_1} follows from choosing a uniform random $\sigma$ revealed to the servers; \eqref{eq:lem:transitivegraphs_2} applies the change of variables $\sigma\mapsto\sigma\circ\tau$; and \eqref{eq:lem:transitivegraphs_3} uses $\tau(i)=j$. This gives us \eqref{eq:transitivegraphsequalentropy}.

Similarly, we explicitly establish the symmetric entropy conditioned on $\mathcal{W} \setminus \{W_{i,j}\}$. Notice that the automorphism $\tau$ maps the shared edge $W_{i,j}$ to itself, meaning $\tau(\mathcal{W} \setminus \{W_{i,j}\}) = \mathcal{W} \setminus \{W_{i,j}\}$. Thus:
\begin{align}
    &H(A_{i}^{\Pi'} \mid \mathcal{W} \setminus \{W_{i,j}\},\mathcal{Q}^{\Pi'}) \nonumber \\
    &= \frac{1}{|\text{Aut}(K_{N})|}\sum_{\sigma\in \text{Aut}(K_{N})} H(A_{\sigma(i)}^{\Pi} \mid \sigma(\mathcal{W} \setminus \{W_{i,j}\}),\mathcal{Q}^{\Pi}) \nonumber \\
    &\hspace{-0.2cm}= \frac{1}{|\text{Aut}(K_{N})|}\sum_{\sigma\in \text{Aut}(K_{N})} H(A_{\sigma(\tau(i))}^{\Pi} \mid \sigma(\tau(\mathcal{W} \setminus \{W_{i,j}\})),\mathcal{Q}^{\Pi}) \nonumber \\
    &= \frac{1}{|\text{Aut}(K_{N})|}\sum_{\sigma\in \text{Aut}(K_{N})} H(A_{\sigma(j)}^{\Pi} \mid \sigma(\mathcal{W} \setminus \{W_{i,j}\}),\mathcal{Q}^{\Pi}) \nonumber \\
    &= H(A_{j}^{\Pi'} \mid \mathcal{W} \setminus \{W_{i,j}\},\mathcal{Q}^{\Pi'}).
    \label{eq:lem:transitivegraphs_4}
\end{align}
Now from Lemma \ref{lemma:subpack_entropy_bound} we have
\begin{equation}
\begin{split}
    &\frac{K}{1+(K-1)e^{-\epsilon}} \Big[ H(A_{i}^{\Pi'} \mid \mathcal{W} \setminus \{W_{i,j}\},\mathcal{Q}^{\Pi'}) \\
    &\quad + H(A_{j}^{\Pi'} \mid \mathcal{W} \setminus \{W_{i,j}\},\mathcal{Q}^{\Pi'}) \Big] \ge L. \label{eq:lem:transitivegraphs_5}
\end{split}
\end{equation}
Combining \eqref{eq:lem:transitivegraphs_4} and \eqref{eq:lem:transitivegraphs_5} we get 
\begin{equation*}
    H(A_{i}^{\Pi'} \mid \mathcal{W} \setminus \{W_{i,j}\},\mathcal{Q}^{\Pi'})\geqslant \frac{L(1+(K-1)e^{-\epsilon})}{2K}.
\end{equation*}

We also show that the new scheme is $\epsilon$-private. Because the sampled automorphism $\sigma$ is revealed to the servers, the privacy condition evaluates on a per-$\sigma$ basis. Let $\{i_1, j_1\}$ and $\{i_2, j_2\}$ represent the edges obtained by mapping the distinct target files $k_1$ and $k_2$ through the automorphism $\sigma$ (i.e., $\sigma(k_1) = \{i_1, j_1\}$ and $\sigma(k_2) = \{i_2, j_2\}$). Now,
\begin{align}
    \frac{P(\mathcal{Q}_{i}^{\Pi_{\sigma}}=q|\theta=k_1)}{P(\mathcal{Q}_{i}^{\Pi_{\sigma}}=q|\theta=k_2)} &= \frac{P(\mathcal{Q}_{\sigma(i)}^{\Pi}=q|\sigma(\theta)=\{i_{1},j_{1}\})}{P(\mathcal{Q}_{\sigma(i)}^{\Pi}=q|\sigma(\theta)=\{i_{2},j_{2}\})} \nonumber \\
    &= \frac{P(\mathcal{Q}_{\sigma(i)}^{\Pi}=q|\theta=\{\sigma^{-1}(i_{1}),\sigma^{-1}(j_{1})\})}{P(\mathcal{Q}_{\sigma(i)}^{\Pi}=q|\theta=\{\sigma^{-1}(i_{2}),\sigma^{-1}(j_{2})\})} \label{eq:lem:transitivegraphs_6} \\
    &\le e^{\epsilon}. \label{eq:lem:transitivegraphs_7}
\end{align}
\eqref{eq:lem:transitivegraphs_6} holds because $\{\sigma^{-1}(i_{1}),\sigma^{-1}(j_{1})\}$ and $\{\sigma^{-1}(i_{2}),\sigma^{-1}(j_{2})\}$ are valid edges in the complete graph $K_{N}$ and \eqref{eq:lem:transitivegraphs_7} follows from the $\epsilon$-privacy of $\Pi$.
\end{IEEEproof}

Building upon the prior lemmas, we now derive an explicit upper bound on the retrieval rate for complete graphs.
\begin{theorem}
    For a complete graph $K_N$ on $N \geq 3$ vertices and $K = N(N-1)/2$, any $\epsilon$-private scheme has a rate satisfying
\label{thm:rate_complete_graphs}
\begin{align*}
    R\leq &~\frac{2K \big(1+(K-1)e^{\epsilon}\big)}{ \big(1+(K-1)e^{-\epsilon}\big)} \\ 
    &\qquad\times\Big[ N + 1 + (K-1)\big(2e^{\epsilon} + (N-1)e^{-\epsilon}\big) \Big]^{-1}.
\end{align*}
\end{theorem}
\begin{IEEEproof} We begin the proof along the same lines of Theorem \ref{thm:rate_general_graphs}. For $N \ge 3$, without loss of generality, assume the target file index is $\theta = k = \{i,j\}$ such that $N \notin \{i,j\}$ (which is always possible by relabeling the servers). Consider a $\epsilon$-private protocol for $K_N$. We wish to bound its rate, or equivalently its expected download. By Lemma \ref{lemma:transitive_graphs}, it is sufficient to consider the equivalent protocol which satisfies the conditions \eqref{eq:transitivegraphsequalentropy} and \eqref{eq:completegraphsubpackentropy}. Now, we have the following. 
\begin{align}
    L &= H(W_{i,j}) = H(W_{i,j}|\mathcal{Q},\theta=k) \label{eq:thm:completegraphs_1} \\
    &= H(W_{i,j}|\mathcal{Q},\theta=k) - H(W_{i,j}|A_{[N]},\mathcal{Q},\theta=k) \label{eq:thm:completegraphs_2} \\
    &= I(W_{i,j};A_{[N]}|\mathcal{Q},\theta=k) \nonumber \\
    &= H(A_{[N]}|\mathcal{Q},\theta=k) - H(A_{[N]}|W_{i,j},\mathcal{Q},\theta=k) \nonumber \\
    &\le \sum_{m=1}^{N}H(A_{m}|\mathcal{Q},\theta=k) \nonumber \\
    &\qquad - H(A_{[N]}|W_{i,j},\mathcal{Q},\theta=k) \nonumber \\
    &\le \frac{K}{1+(K-1)e^{-\epsilon}}\sum_{m=1}^{N}H(A_{m}|\mathcal{Q}) \nonumber \\
    &\qquad - H(A_{[N]}|W_{i,j},\mathcal{Q},\theta=k) \label{eq:thm:completegraphs_3} \\
    &= \frac{NK}{1+(K-1)e^{-\epsilon}}H(A_{1}|\mathcal{Q}) \nonumber \\
    &\qquad - H(A_{[N]}|W_{i,j},\mathcal{Q},\theta=k), \label{eq:thm:completegraphs_4}
\end{align}
where \eqref{eq:thm:completegraphs_1} follows since the queries $\mathcal{Q}$ and the identity of the required file are independent of the required file content; \eqref{eq:thm:completegraphs_2} follows from the correctness condition; \eqref{eq:thm:completegraphs_3} is due to Lemma \ref{lemma:epsilon_privacy_theta_entropy}; and \eqref{eq:thm:completegraphs_4} holds because of \eqref{eq:transitivegraphsequalentropy}. Rearranging the above inequality gives the following:
\begin{align}
& N H(A_{1}|\mathcal{Q}) \nonumber \\
&\quad\ge \frac{1+(K-1)e^{-\epsilon}}{K}\bigg( L + H(A_{[N]}|W_{i,j},\mathcal{Q},\theta=k) \bigg) \nonumber \\
&\quad= \frac{1+(K-1)e^{-\epsilon}}{K} \Big( L \nonumber \\
&\qquad \qquad + \sum_{m=1}^{N}H(A_{m}|A_{[m-1]},W_{i,j},\mathcal{Q},\theta=k) \Big) \label{eq:thm:completegraphs_4_2} \\
&\quad= \frac{1+(K-1)e^{-\epsilon}}{K} \Big( L \nonumber \\
&\qquad \qquad + \sum_{m=1}^{N-1}H(A_{m}|A_{[m-1]},W_{i,j},\mathcal{Q},\theta=k) \Big) \label{eq:thm:completegraphs_5} \\
&\quad\ge \frac{1+(K-1)e^{-\epsilon}}{K} \Big( L \nonumber \\
&\qquad \qquad + \sum_{m=1}^{N-1}H(A_{m}|A_{[m-1]},\mathcal{W}\setminus\{W_{m,N}\},\mathcal{Q},\theta=k) \Big) \label{eq:thm:completegraphs_6} \\
&\quad= \frac{1+(K-1)e^{-\epsilon}}{K} \Big( L \nonumber \\
&\qquad \qquad + \sum_{m=1}^{N-1}H(A_{m}|\mathcal{W}\setminus\{W_{m,N}\},\mathcal{Q},\theta=k) \Big) \label{eq:thm:completegraphs_7} \\
&\quad\ge \frac{1+(K-1)e^{-\epsilon}}{K} \Big( L \nonumber \\
&\qquad \qquad + \frac{K}{1+(K-1)e^{\epsilon}} \sum_{m=1}^{N-1}H(A_{m}|\mathcal{W}\setminus\{W_{m,N}\},\mathcal{Q}) \Big) \label{eq:thm:completegraphs_8} \\
&\quad\ge \frac{1+(K-1)e^{-\epsilon}}{K}\left(L+\frac{L(N-1)(1+(K-1)e^{-\epsilon})}{2(1+(K-1)e^{\epsilon})}\right) \label{eq:thm_comp1:9} \\
\nonumber &=\frac{L \big(1+(K-1)e^{-\epsilon}\big)}{2K \big(1+(K-1)e^{\epsilon}\big)} \\
&\qquad\times \Big[ N + 1 + (K-1)\big(2e^{\epsilon} + (N-1)e^{-\epsilon}\big) \Big], \nonumber
\end{align}

 where \eqref{eq:thm:completegraphs_4_2} follows from the chain rule of entropy; \eqref{eq:thm:completegraphs_5} follows from the non-negativity of entropy; \eqref{eq:thm:completegraphs_6} follows since conditioning reduces entropy (and because $N \notin \{i,j\}$, the target file $W_{i,j}$ is always contained within $\mathcal{W}\setminus\{W_{m,N}\}$); \eqref{eq:thm:completegraphs_7} follows because the answers $A_{[m-1]}$ are deterministic functions of $\mathcal{Q}$ and $\mathcal{W}\setminus\{W_{m,N}\}$ (since the excluded file $W_{m,N}$ connects vertex $m$ and $N$, it is not stored at any server in $[m-1]$); \eqref{eq:thm:completegraphs_8} follows from Lemma \ref{lemma:epsilon_privacy_theta_entropy}; and \eqref{eq:thm_comp1:9} follows from \eqref{eq:completegraphsubpackentropy}.
\end{IEEEproof}

\begin{remark}
    Let $R_{\mathrm{bound}}(\epsilon)$ denote the right-hand side of Theorem~\ref{thm:rate_complete_graphs}. It equals $2/(N+1)$ at $\epsilon=0$ (which is the bound reported by \cite{Sadehetal_TIT_2023_BoundsonPIRgraphs}), but exceeds $1$ (and hence becomes vacuous) already at moderate values of $\epsilon$ (numerically, for $\epsilon$ approximately in the range $0.58$--$0.74$ when $3\le N\le 10$), and tends to $K$ as $\epsilon\to\infty$. This is seen in Fig.~\ref{fig:rateprivacy_converse}, where the converse is truncated at $R=1$, so the valid bound is $\min\{1,R_{\mathrm{bound}}(\epsilon)\}$ (recall that $R\le 1$ trivially). The looseness stems from Lemma~\ref{lemma:epsilon_privacy_theta_entropy}, which relates $H(A_n\mid\mathcal{Q},W_\Omega,\theta=k)$ to the $\theta$-average $H(A_n\mid\mathcal{Q},W_\Omega)$ using only the pairwise ratio bounds $e^{\pm\epsilon}$; the resulting factors $K/(1+(K-1)e^{\pm\epsilon})$ approach $0$ and $K$ as $\epsilon$ grows. 
    The lemma does not exploit that a server in $K_N$ stores only $N-1$ of the $K$ files, so its answer can depend on the desired index only through these files. Incorporating this structure could lead to a non-trivial converse for larger $\epsilon$.
\end{remark}

\section{Special Graphs}
In this section, we specialize the $\epsilon$-private scheme for two specific graph classes: complete graphs and complete bipartite graphs. Taking advantage of the inherent symmetry and structural properties of these specific graph families,establish the explicit bias probability requirements needed to satisfy the $\epsilon$-privacy constraint. We utilize the expected download cost expressions for these specific graph classes previously derived in \cite[Section V]{Vidya_ISITA2026} and evaluate their performance under the $\epsilon$-privacy constraint. Furthermore, we illustrate the achievable rate-privacy trade-offs for both graph classes. For complete graphs, we also plot the achievable rate alongside the rate upper bound established in Theorem \ref{thm:rate_complete_graphs}.

\subsection{\textbf{Complete Graphs} $K_N$}
\begin{itemize}
    \item \textit{Download Cost:} In a complete graph, $K = N(N-1)/2$. Then the expected download cost is
    \begin{align*}
        \mathbb{E}[D_{K_N}(p)] &= (1-p) \nonumber \\
        &\quad + \sum_{s=2}^{N-1}\left[1 - \frac{p^s}{K}\left(K - \frac{(s-1)(2p-1)}{p}\right)\right] \nonumber \\
        &\quad + \left[1 - \frac{p^{N-1}}{K}\left(K - \frac{(N-1)(2p-1)}{p}\right)\right]
    \end{align*}
    and rate $R(p) = 1/\mathbb{E}[D_{K_N}(p)]$.
    \item \textit{Privacy Requirement:} For $N \ge 3$, the bias probability $p$ is bounded by
    \begin{equation*}
        p \le \frac{e^{\epsilon/2}}{1 + e^{\epsilon/2}}.
    \end{equation*}
\end{itemize}

\subsection{\textbf{Complete Bipartite Graphs $K_{N/2,N/2}$}}
We consider the complete bipartite graph with $N/2$ vertices on each side, and all possible edges present. 
\begin{itemize}
    \item \textit{Download Cost:} The expected download cost evaluates to
    {\small
    \begin{equation*}
         \mathbb{E}[D_{K_{N/2,N/2}}(p)] = \frac{N}{2}(1-p) + \frac{N}{2} - p^{\frac{N}{2}-1}\left[1+p\left(\frac{N}{2}-2\right)\right]
    \end{equation*}
    }
    and rate $R(p) = 1/\mathbb{E}[D_{K_{N/2,N/2}}(p)]$.
    \item \textit{Privacy Requirement:} For $N \ge 4$, the bias probability $p$ is bounded by
    \begin{equation*}
        p \le \frac{e^{\epsilon/2}}{1 + e^{\epsilon/2}}.
    \end{equation*}
\end{itemize}

\begin{figure}[H]
    \centering
    \includegraphics[width=\columnwidth]{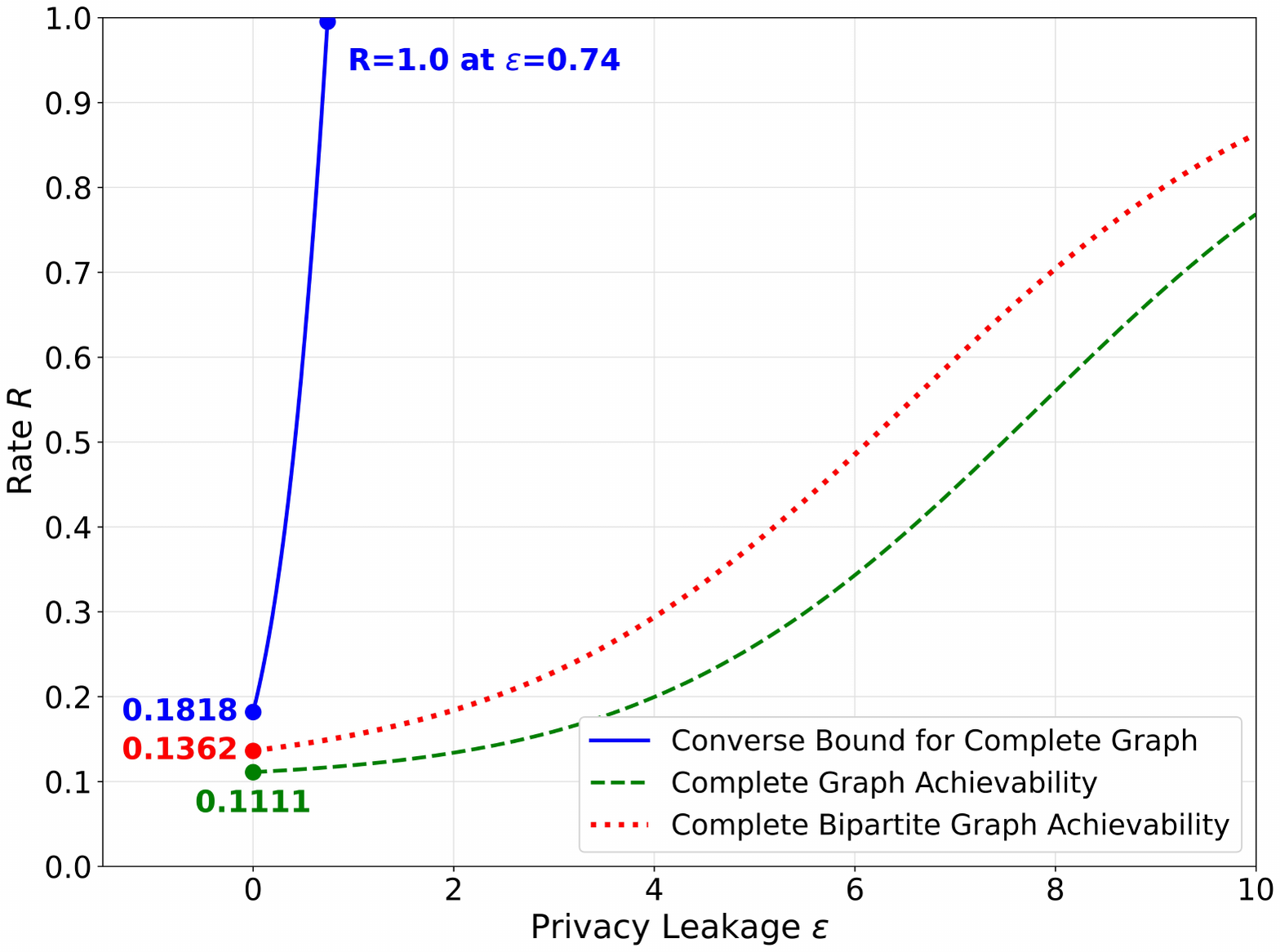}
    \caption{Achievable Rate vs $\epsilon$-Leakage with $N = 10$ for Complete and Complete Bipartite Graphs. The converse bound from Theorem \ref{thm:rate_complete_graphs} is also shown.}
    \label{fig:rateprivacy_converse}
\end{figure}
\section{Conclusion}
In this work, we have studied the problem of Weak-PIR over graph-based storage systems subject to a strict $\epsilon$-privacy constraint. We adapted an independent-set-based PIR protocol for arbitrary graphs and explicitly characterized its pointwise $\epsilon$-leakage guarantees. Furthermore, we developed a framework for bounding the download cost of $\epsilon$-private protocols for arbitrary graphs, which subsequently enabled the derivation of a specialized rate upper bound for complete graphs. We note that this upper bound is informative for small $\epsilon$ only, and in the perfect privacy regime ($\epsilon=0$), it is weaker than the capacity bounds established by Ge et al. \cite{ge2025privateinformationretrievalgraphs}. Finally, we evaluated the performance of the proposed scheme on complete and complete bipartite graphs, illustrating smooth rate-privacy trade-offs. Future research directions include developing tighter converse bounds for general graphs and investigating optimal WPIR protocols for specialized graph classes.

\bibliographystyle{IEEEtran}
\bibliography{arxiv_ref.bib}

\end{document}